\documentclass[bezier]{article} 
\usepackage{times}

\usepackage{amssymb, amsthm, amsmath}
\usepackage{fullpage}
\usepackage{enumerate}
\usepackage[pdftex]{graphicx}
\usepackage{color}

\makeatletter
\newtheorem*{rep@theorem}{\rep@title}
\newcommand{\newreptheorem}[2]{%
\newenvironment{rep#1}[1]{%
 \def\rep@title{#2 \ref{##1}}%
 \begin{rep@theorem}}%
 {\end{rep@theorem}}}
\makeatother

\newtheorem{theorem}{Theorem}[section]
\newreptheorem{theorem}{Theorem}

\newtheorem{lemma}[theorem]{Lemma}

\newtheorem{observation}[theorem]{Observation}

\newcommand{\OPT}{\text{OPT}}
\newcommand{\opt}{\text{OPT}}

\newcommand{\polylog}{\mathop{\mathrm{polylog}}}
\newcommand{\piin}{\pi^{\text{in}}}
\newcommand{\piout}{\pi^{\text{out}}}
\newcommand{\dpt}{\text{DP}}

\newcommand{\leng}{\text{length}}
\newcommand{\depth}{\text{depth}}
\newcommand{\Prob}{\text{Prob}}
\newcommand{\aalpha}{A}
\newcommand{\ssigma}{B}
\newcommand{\rrho}{D}

\newenvironment{smquotehelper}{\begin{list}{}{\setlength{\itemsep}{.25ex}\setlength{\labelsep}{0pt}
\setlength{\topsep}{.25ex}\setlength{\leftmargin}{3em}
\setlength{\parsep}{.1ex}}}{\end{list}}

\newif\ifabstract
\abstractfalse
\newif\iffull
\ifabstract \fullfalse \else \fulltrue \fi

\begin{document}

\title{A polynomial-time approximation scheme for Euclidean Steiner
  forest\footnote{This version is more recent than that appearing in
    the FOCS proceedings.  The partition step has been corrected and
    the overall presentation has been clarified and formalized. This material is based upon
work supported by the National Science Foundation under Grant Nos.\ CCF-0635089, 
CCF-0964037, and CCF-0963921.}}

\author{Glencora Borradaile\\
School of EECS\\Oregon State University\\glencora@eecs.orst.edu
  \and 
Philip N. Klein\thanks{Work done while visiting MIT's CSAIL.} 
\\Computer Science\\Brown
  University\\klein@brown.edu 
\and 
Claire Mathieu \\CNRS\\ Ecole Normale Sup\'erieure\\cmathieu@di.ens.fr }

\date{}
\maketitle
\thispagestyle{empty}

\begin{abstract}
  We give a randomized $O(n \polylog n)$-time approximation scheme for
  the Steiner forest problem in the Euclidean plane.  For every fixed
  $\epsilon > 0$ and given $n$ terminals in the plane with connection
requests between some pairs of terminals,
  our scheme finds a $(1+\epsilon)$-approximation to the
  minimum-length forest that connects every requested pair of terminals.
\end{abstract}

\section{Introduction}

\subsection{Result and background}

In the Steiner forest problem, we are given a set of $n$ pairs of {\em
  terminals} $\{(t_i,t_i')\}_{i = 1}^n$.  The goal is to find a
minimum-cost forest $F$ such that every pair of terminals is connected
by a path in $F$.  We consider the problem where the terminals are
points in the Euclidean plane.  The solution is a set of line segments
of the plane; non-terminal points with more than two line segments
adjacent to them in the solution are called {\em Steiner points}.  The
cost of $F$ is the sum of the lengths in $\ell_2$ of the line
segments comprising it.  Our main result is:
\begin{theorem}\label{thm:main}
  There is a randomized $O(n \polylog n)$-time approximation
  scheme for the Steiner forest problem in the Euclidean plane.
\end{theorem}
An approximation scheme is guaranteed, for a fixed $\epsilon$, to find a solution whose total length is an most $1+\epsilon$ times the length of a minimum solution.

Independently, Mitchell~\cite{Mitchell99} and Arora~\cite{Arora98}
developed a method for designing polynomial-time approximation schemes
(PTASes) for problems such as such as Traveling Salesman and Steiner
tree in the Euclidean plane.  The running time for Arora's technique
was improved upon by Rao and Smith for the Steiner tree and TSP
problems~\cite{RS98} and others have extended these techniques to give
PTASes for other problems, e.g., k-medians~\cite{ARR98,KR07}.  Our work
builds on the these techniques, using the framework as described by
Arora.

The Steiner forest problem, a generalization of the Steiner tree
problem, is NP-hard~\cite{Karp75} and max-SNP
complete~\cite{BP89,Thimm01} in general graphs and high-dimensional
Euclidean space~\cite{Trevisan01}.  Therefore, no PTAS exists for
these problems.  The 2-approximation algorithm due to Agrawal, Klein
and Ravi~\cite{AKR95} can be adapted to Euclidean problems by
restricting the Steiner points to lie on a sufficiently fine grid and
converting the problem into a graph problem.  

We have formulated the connectivity requirements in terms of {\em
  pairs} of terminals.  One can equivalently formulate these in terms of
{\em sets} of terminals: the goal is then to find a forest in which
each set of terminals are connected.  Arora states~\cite{Arora2003}
that his approach yields an approximation scheme whose running time is
exponential in the number of sets of terminals, and this is the only
previous work to take advantage of the Euclidean plane to get a better
approximation ratio than that of Agrawal et al.~\cite{AKR95}.

\subsection{Recursive dissection}\label{subsection:dissection}

In Arora's paradigm, the feasible space is recursively decomposed by
{\em dissection squares} using a randomized variant of the quadtree
(Figure~\ref{fig:dissections}).  The dissection is a 4-ary tree whose
root is a square box enclosing the input terminals, whose width $L$ is
twice the width of the smallest square box enclosing the terminals,
and whose lower left-hand corner of the root box is translated from
the lower left-hand corner of the bounding box by $(-a,-b)$, where $a$
and $b$ are chosen uniformly at random from the range $[0,L/2)$. Each
node in the tree corresponds to a {\em dissection square}.  Each
square is dissected into four parts of equal area by one
vertical and one horizontal {\em dissection line} each spanning the
breadth of the root box.  This process continues until each
square contains at most one terminal (or multiple terminals having the
same coordinates). 

\begin{figure}[h]
  \centering
  \includegraphics[scale = 2]{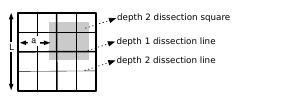}
  \caption{The shifted quad-tree dissection.  The shaded box is
  the bounding box of the terminals.}
  \label{fig:dissections}
\end{figure}

Feasible solutions are restricted to using a small number of {\em
  portals}, designated points on each dissection line.  A Structure
Theorem states that there is a near-optimal solution that obeys these
restrictions.  The final solution is found by a dynamic program guided
by the recursive decomposition.

In the problems considered by Arora, the solutions are connected.
However, the solution to a Steiner forest problem is in general
disconnected, since only paired terminals are required to be
connected.  It is not known {\em a priori} how the connected
components partition the terminal pairs.  For that reason, maintaining
feasibility in the dynamic program requires a table that is is
exponential in the number of terminal pairs.  In fact, Arora
states~\cite{Arora2003} that his approach yields an approximation
scheme whose running time is exponential in the number of sets of
terminals.

Nevertheless, here we use Arora's approach to get an approximation
scheme whose running time is polynomial in the number of sets of
terminals.  The main technical challenge is in maintaining feasibility in a small dynamic programming table.

\subsection{Small dynamic programming table} \label{sec:overview}

We will use Arora's approach of a random recursive dissection.  Arora
shows (ie.~for Steiner tree) that the optimal solution can be
perturbed (while increasing the length only slightly) so that, for
each box of the recursive dissection, the solution within the box
interacts weakly and in a controlled way with the solution outside the
box.  In particular, the perturbed solution crosses the boundary of
the box only a constant number of times, and only at an $O(1)$-sized
subset of $O(\log n)$ selected points, called {\em portals}.  The
optimal solution that has this property can be found using dynamic
programming.

Unfortunately, for Steiner forest those restrictions are not
sufficient: maintaining feasibility constraints cannot be done with a
polynomially-sized dynamic program. To see why, suppose the solution
uses only 2 portals between adjacent dissection squares $R_E$ and
$R_W$. In order to combine the solutions in $R_W$ and $R_E$ in the
dynamic program into a feasible solution in $R_W\cup R_E$, we need to
know, for each pair $(t,t')$ of terminals with $t \in R_W$ and $t' \in
R_E$, which portal connects $t$ and $t'$ (Figure~\ref{fig:prop5}(a)).
This requires $2^n$ configurations in the dynamic programming table.

\begin{figure}[h]
  \centering
  \includegraphics[scale=2]{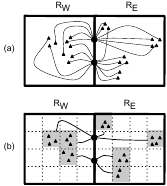}
  \caption{Maintaining feasibility is not trivially polynomial-sized.}
  \label{fig:prop5}
\end{figure}

To circumvent the problem in this example, the idea is to decompose
$R_W$ and $R_E$ into a constant number of smaller dissection squares
called {\em cells}.  All terminals in a common cell that go to the
boundary use a common portal.  Thus, instead of keeping track of
each terminal's choice of portal individually, the dynamic program can
simply memoize each cell's choice of portal.  The dynamic program
also uses a specification of how portals must be connected {\em outside} the
dissection squares.  This information is sufficient to check feasibility when
combining solutions of the subproblems for $R_W$ and for $R_E$.
To show near-optimality, we show that a constant number of cells per square is sufficient for finding a nearly-optimal solution.

\paragraph{Basic notation and definitions}

For two dissection squares $A$ and $B$, if $A$ encloses $B$, we say
that $B$ is a {\em descendent} of $A$ and $A$ is an {\em ancestor} of
$B$.  If no other dissection square is enclosed by $A$ and encloses
$B$, we say that $A$ is the {\em parent} of $B$ and $B$ is the {\em
  child} of $A$.  We will extend these definitions to describe
relationships between cells.  The {\em depth} of a square $S$ is given
by its depth in the dissection tree ($0$ for the root).  The depth of a dissection line
is the minimum depth of squares it separates.   Note that a square at depth $i$ is bounded
by two perpendicular depth-$i$ lines and two lines of depth less than
$i$.

For a line segment $s$ (open or closed), we use $\leng(s)$ to denote
the $\ell_2$ distance between $s$'s endpoints.  For a set of line
segments $S = \{s_1, s_2, \ldots\}$, $\leng(S) = \sum_i \leng(s_i)$.
For a subset $X$ of the Euclidean plane, a component of $X$ is a
maximal subset $Y$ of $X$ such that every pair of points in $Y$ are
path-wise connected in $X$.  We use $|X|$ to denote the number of
components of $X$.  \newcommand{\diam}{\text{diam}} The diameter of a
connected subset $C$ of the Euclidean plane, $\diam(C)$, is the maximum $\ell_2$ distance between any pair of points in $C$.  We use $\opt$ to denote both the line segments forming an optimal solution and the length of those line segments.

\section{The algorithm}

The algorithm starts by finding a rough {\em partition} of the terminals which is a coarsening of the 
connectivity requirements (subsection~\ref{subsection:step1}).  We solve each part of this partition independently.  
We next {\em discretize} the problem by moving the terminals to integer coordinates of a sufficiently fine grid 
(subsection~\ref{subsection:step2}).  We will also require that the Steiner points be integer coordinates.  
We next perform a recursive dissection (subsection~\ref{subsection:step3}) and assign points on the dissection lines as portals
 (subsection~\ref{subsection:step4}) as introduced in Section~\ref{subsection:dissection}.  We then break each dissection square into a small number of cells.  We find the best feasible solution $F$ to the
discretized problem that only crosses between dissection squares
at portals and such that for each cell $C$ of dissection square $R$,
$F \cap R$ has only one component that connects $C$ to the boundary of
$R$ (subsection~\ref{subsection:step5}).  

We will show that the expected length of $F$ is at most a
$\frac{4}{10}\epsilon$ fraction longer than $\OPT$.  By Markov's inequality, with probability at least
one-half the $\leng(F)\leq (1+\frac{8}{10}\epsilon) \OPT$.  We show
that by moving the terminals back to their original positions (from
their nearest integer coordinates) increases the length by at most
$\frac{\epsilon}{40}\OPT$.  Therefore, the output solution has length
at most $(1+\epsilon)\OPT$ with probability one half.

We now describe each of these steps in detail.

\subsection{Partition}\label{subsection:step1}

We first partition the set of terminal pairs, creating subproblems
that can be solved independently of each other without loss of
optimality.  The purpose of this partition is to bound the size of the
bounding box for each problem in terms of $\OPT$.  This bound is
required for the next step, the result of which allows us to treat
this geometric problem as a combinatorial problem.  This
discretization was also key to Arora's scheme, but the bound on
the size of the bounding box for the problems he considers is
trivially achieved.  This is not the case for the Steiner forest
problem.  The size of the bounding box of all the terminals in an
instance may be unrelated to the length of $\OPT$.

\newcommand{\db}{\text{dist}} Let $Q$ be the set of $m$ pairs
$\{(t_i,t'_i)\}_{i=1}^m$ of $n$ terminals.  Consider the Euclidean
graph whose vertices are the terminals and whose edges are the line
segments connecting terminal pairs in $Q$ and let $C_1, C_2, \ldots$
be the components of this graph.  Let $\db(Q)=\max_i \diam(C_i)$; this
is the maximum distance between any pair of terminals that must be
connected.

\begin{theorem}\label{thm:modify2}
  There exists a partition of $Q$ into independent instances $Q_1,
  Q_2, \ldots$ such that the optimal solution for $Q$ is the disjoint
  union of optimal solutions for each $Q_i$ and such that the diameter
  of $Q_i$ is at most $n_i^2 \db(Q_i)$ where $n_i$ is the number of
  terminals in $Q_i$.  Further, this partition can be found in
  $O(n\log n)$ time.
\end{theorem}

\noindent We will show that the following algorithm, {\sc
  Partition}$(Q)$, produces such a partition. Let $T$ be the minimum spanning tree of the terminals in $Q$ .
\begin{tabbing}
  {\sc Partition}$(Q,T)$\\
  \qquad \= Let $e$ be the longest edge of $T$. \\
  \> If $\leng(e) > n\,\db(Q)$, \\
  \> \qquad \= remove $e$ from $T$ and let $T_1$ and $T_2$ be the
  resulting components.\\
  \> \> For $i = 1, 2$, let $Q_i$ be the subset of terminal pairs connected by $T_i$.\\
  \> \> $T:= \text{\sc Partition} (Q_1,T_1) \cup \text{\sc Partition} (Q_2,T_2)$.\\
  \> Return the partition defined by the components of $T$.
\end{tabbing}

\begin{proof}[Proof of Theorem~\ref{thm:modify2}]
  First observe that by the cut property of minimum spanning trees, the
  distance between every terminal in $T_1$ and every terminal in $T_2$
  is at least as long as the edge that is removed.
  
  Since a feasible solution is given by the union of minimum spanning
  trees of the sets of the requirement partition, and each edge in
  these trees has length at most $\db(Q)$, $\OPT < n\, \db(Q)$.
  $\OPT$ cannot afford to connect a terminal of $T_1$ to a terminal of
  $T_2$, because the distance between any terminal in $T_1$ and any
  terminal in $T_2$ is at least $n\, \db(Q)$ which is greater than the
  lower bound. (By definition of $\db$, there cannot have been a
  requirement to connect a terminal of $T_1$ to a terminal of $T_2$.)
  Therefore, $\OPT$ must be the union of two solutions, one for the
  terminals contained by $T_1$ and one for the terminals contained by
  $T_2$.  Inductively, the optimal solution for $Q$ is the union of
  optimal solutions for each set in {\sc Partition}$(Q)$, giving the
  first part of the theorem.

  The stopping condition of {\sc Partition} guarantees that there is a 
  spanning tree of the terminals in the current subset $Q_i$ of terminals
  whose edges each have length at most $n_i\,\db(Q_i)$.  Therefore,
  there is a path between each pair of terminals of length at most
  $n_i^2\,\db(Q_i)$, giving the second part of the theorem.

  Finally, we show that {\sc Partition} can be implemented to run in
  $O(n \log n)$ time. The diameter of a set of points in the Euclidean
  plane can be computed by first finding a convex hull and this can be
  done in $O(n\log n)$ by, for example, Graham's
  algorithm~\cite{Graham72}.  Therefore, $\db(C_i)$ can be computed in
  $O(n \log n)$ time.  The terminal-pair sets $Q_1$ and $Q_2$ for the
  subproblems need not be computed explicitly as the required
  information is given by $T_1$ and $T_2$.  By representing $T$ with a
  top-tree data structure, we can find $n_i$ and $d(Q_i)$ by way of a
  cut operation and a sum and maximum query, respectively, in $O(\log
  n)$ time~\cite{GGT91}.  Since there are $O(n)$ recursive calls, the
  total time for the top-tree operations is $O(n \log n)$.
\end{proof}

Our PTAS finds an approximately optimal solution to each subproblem
$Q_i$ (as defined by Theorem~\ref{thm:modify2}) and combines the
solutions.  For the remainder of our description of the algorithm, we
focus on how the algorithm addresses one such subproblem $Q_i$.  In
order to avoid carrying over subscripts and arguments $Q_i$,
$\db(Q_i)$, $n_i$ throughout the paper, from now on we will consider
an instance given by $Q$, $\db(Q)$, and $n$, and assume it has the
property that the maximum distance between terminals, whether
belonging to a requirement pair or not, is at most $n^2 \db(Q)$.  $\OPT$
will refer to the length of the optimal solution for this subproblem.

\subsection{Discretize}\label{subsection:step2}

We would like to treat the terminals as discrete combinatorial objects.  In
order to do so, we assume that the coordinates of the
terminals lie on an integer grid.  We can do so by {\em scaling} the
instance, but this may result in coordinates of unreasonable size.
Instead, we scale by a smaller factor and {\em round} the positions
of the terminals to their nearest half-integer coordinates.

\subsubsection*{Scale} 

We scale by a factor of \[\frac{40\sqrt 2 n}{\epsilon\, \db(Q)}.\]  Before scaling,
$\OPT \ge \db(Q)$, the distance between the furthest pair
of terminals that must be connected. After scaling we get the
following lower bound:
\begin{equation}
  \label{eq:OPT-lb}
  \OPT \geq \frac {40 \sqrt 2 n} {\epsilon}
\end{equation}
Before scaling, $\diam(Q) \le n^2\,\db(Q)$ by Theorem~\ref{thm:modify2}. After scaling we get the following upper
bound on the diameter of the terminals:
\begin{equation}
  \label{eq:diam-ub}
  \diam(Q) \leq \frac {40 \sqrt 2 n^3} {\epsilon}
\end{equation}
Herein, $\OPT$ refers to distances in the scaled version.

\subsubsection*{Round} 

We round the position of each terminal to the nearest grid center.
Additionally, we will search for a solution that only uses Steiner
points that are grid centers.  We call this constrained problem the
{\em rounded problem}.  The rounded problem may merge terminals (and
thus, their requirements).

\begin{lemma}\label{lemma:rounding}  
  A solution to the Steiner forest can be derived from an optimal
  solution to the rounded problem at additional cost at most 
  $\frac \epsilon {40} \OPT$.
\end{lemma}

\begin{proof}
  Let $F$ be an optimal solution to the rounded problem.  From this we
  build a solution to the original problem by connecting the original
  terminals to their rounded counterparts with line segments of length
  at most $1/ \sqrt 2$, ie.~half the length of the diagonal of a unit
  square.  There are $n$ terminals, so the additional length is at
  most $n/\sqrt 2$ which is at most ${\epsilon \over 40} \OPT$ by
  Equation~\eqref{eq:OPT-lb}.
\end{proof}

Let $F$ be an optimal solution to the rounded problem.  We relate the number of intersections of $F$ with grid lines to $\leng(F)$.  We will bound the
cost of our restrictions to portals and cells with this relationship.

\begin{lemma} \label{lem:sum-of-crossings} There is a solution to the
  rounded problem of length $(1+\frac{1}{10}\epsilon)\OPT$ that
  satisfies
  \begin{equation} \label{eq:crossings-vs-length}
      \sum_{\text{grid lines }\ell} |F \cap \ell| \leq 3 \OPT.
  \end{equation}
\end{lemma}

\begin{proof}
  We build a solution $F$ to the rounded problem from $\OPT$ by
  replacing each line segment $e$ of $\OPT$ with a line segment $e'$ that connects
  the half-integer coordinates that are nearest $e$'s endpoints
  (breaking ties arbitrarily but consistently).  Since the additional length needed for this transformation is at most twice (for each
  endpoint of $e$) the distance from a point to the nearest
  half-integer coordinate:
\[
\leng(e') \le \leng(e)+\sqrt 2
\]
Since $\OPT$ has at most $n$ leafs, $\OPT$ has fewer than $n$ Steiner points
and so has fewer than $4n$ edges.  The additional length is therefore
no greater than $4\sqrt 2 n$.  Combining with
Equation~\eqref{eq:OPT-lb}, this is at most ${1 \over 10}\epsilon
\OPT$.

  $F$ is composed of line segments whose endpoints are half-integer
  coordinates.  Such a segment $S$ of length $s$ can cross at most $s$
  horizontal grid lines and at most $s$ vertical grid lines.
  Therefore
  \[
  \sum_{\text{grid lines }\ell} |S \cap \ell| \leq 2s
  \]
  and summing over all segments of $F$ gives
  \[
  \sum_{\text{grid lines }\ell} |F \cap \ell| \leq 2\,\leng(F) \leq 2(1+{1
    \over 10}\epsilon)\OPT < 3\OPT
  \]
  where the last inequality follows from $\epsilon < 1$.
\end{proof}

From here on out, our goal is to find the solution that is guaranteed
by Lemma~\ref{lem:sum-of-crossings}.  We will not be able to find this
solution optimally, but will be able to find a solution within our error bound of $\epsilon\, \opt$.

\subsection{Dissect}\label{subsection:step3}

The recursive dissection starts with an $L \times L$ box that encloses
the terminals and where $L$ is at least twice as big as needed.  This
allows some choice in where to center the enclosing box.  We make this
choice randomly.  This random choice is used in bounding the incurred
cost, in expectation, of structural assumptions
(Section~\ref{sec:cell-props}) that help to reduce the size of the
dynamic programming table.

Formally, let $L$ be the smallest power of $2$ greater than $2
\cdot\text{diameter}(Q)$. In combination with
Equation~\eqref{eq:diam-ub}, we get the following upper bound on $L$:
\begin{equation}
L \leq { 160 \sqrt 2 \over \epsilon} n^3 \label{eq:L}
\end{equation}
The $x$-coordinate (and likewise the $y$-coordinate) of the lower left
corner of the enclosing box are chosen uniformly at random from the
$L/2$ integer coordinates that still result in an enclosing box.  We
will refer to this as the {\em random shift}.

As described in section~\ref{subsection:dissection}, we perform a
recursive dissection of this enclosing box.  This can be done in $O(n
\log n)$ time~\cite{BET93}.  By our choice of $L$ and the random
shift, this dissection only uses the grid lines. Since the recursive
dissection stops with unit dissection squares, the quad-tree has depth
$\log L$.

Consider a vertical grid line $\ell$.  Since there are $L/2$ values of
the horizontal shift, and $2^{i-1}$ of these values will result in $\ell$
being a depth-$i$ dissection line, we get
\begin{equation}\label{eq:prob}
  \Prob[\depth(\ell)=i] = 2^i/L
\end{equation}

\subsection{Designate portals}\label{subsection:step4}

We designate a subset of the points on each dissection line as {\em
  portals}.  We will restrict our search for feasible solutions that
cross dissection lines at portals only.  We use the portal constant
$\aalpha$, where
\begin{equation}
  \label{eq:inter-portal-distance}
  \aalpha \text{ is the smallest power of two greater than }30\epsilon^{-1}\log L.
\end{equation}
Formally, for each vertical (resp. horizontal) dissection line $\ell$, we
designate as portals of $\ell$ the points on $\ell$ with $y$-coordinates
(resp. $x$-coordinates) which are integral multiples
of 
\[\frac{L}{\aalpha 2^{\text{depth}(\ell)}}.\]
There are no portals on the sides of root dissection square, the
bounding box.  Since a square at depth $i$ has sidelength $L/2^i$ and is bounded by 4 dissection lines at depth at most $i$, we get:
\begin{lemma} \label{lem:n-square-portals}
  A dissection square has at most $4\aalpha$ portals on its boundary.
\end{lemma}

Consider perpendicular dissection lines $\ell$ and $\ell'$.  A portal $p$ of $\ell$ may happen to be a point of $\ell'$ (namely, the intersection point), but $p$ may not be a portal of $\ell'$, that is, it may not be one of the points of $\ell'$ that were designated according to the above definition.

The following lemma will be useful in Subsection~\ref{subsection:establishing-the-portal-property} for technical reasons.

\begin{lemma} \label{lem:corner-portals} For every dissection square
  $R$, the corners of $R$ are portals (except for the points that are
  corners of the bounding box).
\end{lemma}

\begin{proof} Consider a square $R$ at depth $i$.  
Consider the two dissection lines that divide $R$ into 4 $\ell$ and $\ell'$. 
The depth of these lines is $i+1$.  These lines restricted to $R$, namely 
$\ell_R = \ell \cap R$ and $\ell_R' = \ell' \cap R$, have length $L/2^i$, a power of 2.  
Portals are designated as integral multiples of $L/(2^{i+1} \aalpha)$, also a power of 2 and a $1/2\aalpha$ 
fraction of the length of $\ell_R$ and $\ell_R'$.  
It follows that the endpoints and intersection point of $\ell_R$ and $\ell_R'$ are portals of these lines.  
\end{proof}

\subsection{Solve via dynamic programming}\label{subsection:step5}\label{section:dynamicprogram}

In order to overcome the computational difficulty associated with maintaining feasibility (as illustrated in Figure~\ref{fig:prop5}), we divide each dissection square $R$ into a regular
$\ssigma \times \ssigma$ grid of {\em cells}; $\ssigma$, which will be
defined later, is $O(1/\epsilon)$ and is a power of 2.  Each {\em
  cell} of the grid is either coincident with a dissection square or
is smaller than the leaf dissection squares.  Consider parent and child dissection squares $R_P$ and $R_C$; a cell $C$ of $R_p$ encloses four cells of $R_C$.

The dynamic programming table for a dissection square $R$ will be indexed by two subpartitions (partition of a subset) of the portals and cells of $R$; one subpartition will encode the connectivity achieved by a solution within $R$ and the other will encode the connectivity required by the solution outside $R$ in order to achieve feasibility.  The details are given in the next section.

\section{The Dynamic Program} \label{sec:DP}

\subsection{The dynamic programming algorithm}

The dynamic program will only encode subsolutions that have low complexity and permit feasibility.  We call such subsolutions {\em conforming}.  We build a dynamic programming table for each dissection square.  The table is indexed by {\em valid configurations} and the entry will be the best {\em compatible} conforming subsolution.

\subsubsection*{Low complexity and feasible: conforming subsolutions}
Let $R$ be a dissection square or a cell, and let $F$ be a finite
number of line segments of $R$.   
We say that $F$ {\em
  conforms} to $R$ if it satisfies the following properties:
\begin{itemize}
\item ({\em boundary property}) $|F \cap \partial R| \le 4(\rrho+1)$.
\item ({\em portal property}) Every connected component of
  $F\cap \partial R$ contains a portal of $R$.
\item ({\em cell property}) Each cell $C$ of $R$ intersects at most
  one connected component of $F$ that also intersects $\partial R$.
\item ({\em terminal property}) If a terminal $t \in R$ is not connected to its mate by $F$ then it is connected to $\partial R$ by $F$.
\end{itemize}
The constant $\rrho$ is defined in Equation~(\ref{eq:rho}) and is
$O(1/\epsilon)$.  Note that the first three properties are those that
bound the complexity of the allowed solutions and the last guarantees
feasibility.  We say that a solution $F$ recursively conforms to $R$ if
it conforms to all descendents dissection squares of $R$ (including
$R$).  We say that a solution $F$ is conforming if it recursively 
conforms to the root dissection square with every terminal connected
to its mate.
It is a trivial corollary of the last property that a conforming
solution is a feasible solution to the Steiner forest problem.  We
will restate and prove the following in Section~\ref{sec:struct}; the
remainder of this section will give a dynamic program that finds a
conforming solution.
\begin{theorem}[Structure Theorem]\label{thm:structure} 
  There is a conforming solution that has, in expectation over the random
  shift of the bounding box, length at most $(1+\frac \epsilon 4)\OPT$.
\end{theorem}

\subsubsection*{Indices of the dynamic programming table: valid configurations} 
The dynamic programming table $\dpt_R$ for a dissection square $R$
will be indexed by subpartitions of the portals and cells of $R$ that
we call {\em configurations}.  A {\em configuration of $R$} is a pair
$(\piin, \piout)$ with the following properties: $\piin$ is a
subpartition of the cells and portals of $R$ such that each part
contains at least one portal and at least one cell; $\piout$ is a
coarsening of $\piin$.  See Figure~\ref{fig:config}. $\piin$ will
characterize the behaviour of the solution inside $R$ while $\piout$
will encode what connections remain to be made in order to make the
solution feasible. For a terminal $t\in R$, we use $C_R[t]$ to denote
the cell of $R$ that contains $t$.  We say a configuration is {\em
  valid} if it has the following properties:
\begin{itemize} 
\item ({\em compact}) $\piin$ has at most $4(\rrho+1)$ parts and contains at most $4(\rrho+1)$ portals.
\item ({\em connecting}) For every terminal $t$ in $R$ whose mate is
  not in $R$, $C_R[t]$ is in a part of $\piin$.  For every pair of
  mated terminals $t,t'$ in $R$, either $C_R[t]$ and $C_R[t']$ are in
  the same part of $\piout$ or neither $C_R[t]$ nor $C_R[t']$ are in
  $\piin$.
\end{itemize}
The connecting property will allow us to encode and guarantee feasible
solutions.  Since a dissection square has $4\aalpha$ portals
(Lemma~\ref{lem:n-square-portals}) and $\ssigma^2$ cells, the first
property bounds the number of configurations:
\begin{lemma}\label{lem:n-configs}
  There are at most $(4\aalpha+\ssigma^2)^{O(\rrho)}$ or
  $(\epsilon^{-2}\log n)^{O(1/\epsilon)}$ compact configurations of a
  dissection square.
\end{lemma}

We will use the following notation to work with configurations:
For a subpartition $\pi$ of $S$ and an element $x\in S$, we use
$\pi[x]$ to denote the part of $\pi$ containing $x$ if there is one,
and $\emptyset$ otherwise. For two subpartitions $\pi$
and $\pi'$ of a set $S$, we use $\pi \vee \pi'$ to denote the finest
possible coarsening of the union of $\pi$ and $\pi'$.  If we eliminate
the elements that are in partition $\pi'$ but not in partition $\pi$,
$\pi \vee \pi'$ is a coarsening of $\pi'$ and vice versa.

\begin{figure}[h]
  \centering
  \resizebox{6cm}{!}{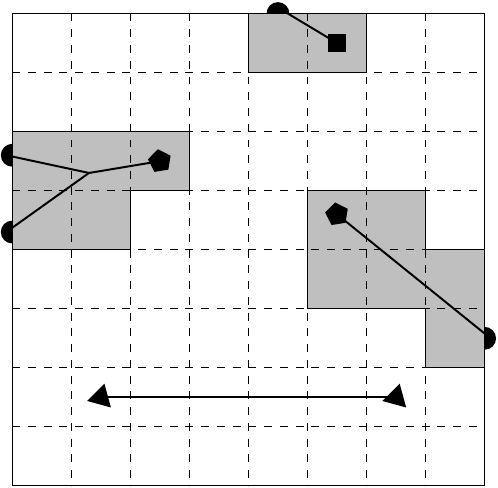}
  \caption{A dissection square and cells (grid), terminal pairs
    (triangles and pentagons) and unmated terminal (square), and
    subsolution (dark lines).  The grey components give the parts of
    $\piin$ with portals (half-disks).  To be a valid configuration,
    the two parts containing the pentagon terminals must be in the
    same part of $\piout$.  The subsolution conforms to $R$ and is
    compatible with $(\piin,\piout)$.  }
  \label{fig:config}
\end{figure}

\subsubsection*{Entries of the dynamic programming table: compatible subsolutions} The entries of
the dynamic programming table for dissection square $R$ are compatible
subsolutions, subsolutions that satisfy.  Formally, a subsolution $F$
and configuration $(\piin, \piout)$ of $R$ are {\em compatible} if and
only if $\piin$ has one part for every connected component of $F$ that
intersects $\partial R$ and that part consists of the cells and
portals of $R$ intersected by that connected component (Figure~\ref{fig:config}).  Note that as
a result, some valid configurations will not have a compatible
subsolution: if a part of $\piin$ contains disconnected cells with
terminals inside, then no set of line segments can connect these
terminals and be contained by the cells of that part.  The entries
corresponding to such configurations will indicate this with $\infty$.
 \begin{observation}\label{obs:conform}
  If $F$ conforms to $R$ then $(\piin, \piout)$ is a
  valid configuration.
\end{observation}
As is customary, our dynamic program finds the {\em value} of the
solution; it is straightforward to augment the program so that the
solution itself can be obtained.  Our procedure for filling the
dynamic programming tables, {\sc populate}, will satisfy the following
theorem:
\begin{theorem} \label{thm:correct} {\sc populate}$(R)$ returns a
  table $\dpt_R$ such that, for each valid configuration
  $(\piin,\piout)$ of $R$, $\dpt_R[\piin, \piout]$ is the minimum
  length of subsolution that recursively conforms to $R$ and that is
  compatible with $(\piin, \piout)$.
\end{theorem}
We prove this theorem in Section~\ref{sec:dp-correct}.

\subsubsection*{Consistent configurations} 
A key step of the dynamic program is to correctly match up the subsolutions of the child dissection squares $R_1, \ldots, R_4$ of $R_0$.  Consider valid configurations $(\piin_i, \piout_i)$ for $i=0,\ldots,4$ and let $\pi^\vee_0 = \bigvee_{i=1}^4 \piin_i$.  We say that the configurations $(\piin_i, \piout_i)$ for $i=0,\ldots,4$ are {\em consistent} if they satisfy the following connectivity requirements:
\begin{enumerate}
\item ({\em internal}) 
$\piin_0$ is given by $\pi^\vee_0$ with portals of $R_i$ that are not portals of $R_0$ removed, parts that do not contain portals of $R_0$ removed,
 and each cell of $R_i$ replaced by the corresponding (parent) cell of $R_0$.
(If non-disjoint parts result from replacing cells by their parents, then the result is not a partition and cannot be $\piin_0$.) 
\item ({\em external}) For two elements (cells and/or portals) $x,x'$ of $R_i$, $\piout_i[x] = \piout_i[x']$ if and only if $\pi^\vee_0[x] = \pi^\vee_0[x']$ or there are portals $p,p'$ such that $\piout_0[p] = \piout_0[p']$, $\pi^\vee_0[x] = \pi^\vee_0[p]$, and $\pi^\vee_0[x'] = \pi^\vee_0[p']$.
\item ({\em terminal}) For mated terminals $t\in R_i$ and $t'\in R_j$ with $1 \le i < j \le 4$, either $\pi^\vee_0[C_i[t]] = \pi^\vee_0[C_j[t']]$  or $\piout_0[C_0[t]] = \piout_0[C_0[t']]$.
\end{enumerate}

\subsubsection*{Dynamic programming procedure} We now give the procedure {\sc populate}
that fills the dynamic programming tables.  The top dissection square
$R$ has a single entry, the entry corresponding to the configuration
$(\emptyset, \emptyset)$.  The desired solution is therefore given by $\dpt_R[\emptyset, \emptyset]$ after filling the table $\dpt_R$ with {\sc populate}$(R)$. 
 The corresponding solution is  conforming.  The following procedure is used to populate the entries of $\dpt_{R_0}$.  The procedure is well defined when the tables are filled for dissection squares in bottom-up order.

\begin{tabbing}
  {\sc populate}$(R_0)$\\
  \qquad\=\qquad\=\qquad\=\qquad\=\qquad\=\qquad\=\hspace{5cm}\=\\
  \> If $R_0$ contains at most one terminal, then \>\>\>\>\>\> {\em \% $R_0$ is a leaf dissection square}\\
  \>\> For every valid configuration $(\piin, \piout)$ of $R_0$,\\
  \>\>\> $\dpt_{R_0}[\piin,\piout] := 0$ \\
  \>\>\> For every part $P$ of $\piin$,\\
  \>\>\>\> if the cells of $P$ are connected and contain the portals (and terminal) of $P$, \\
  \>\>\>\>\> $F_P := $ \=minimum-length set of lines in the cells of $P$ that\\
  \>\>\>\>\>\> connects the portals in $P$ (and terminal, if in $P$), \\
  \>\>\>\> \>$\dpt_{R_0}[\piin,\piout] := \dpt_{R_0}[\piin,\piout] + \leng(F_P)$;\\
  \>\>\>\> otherwise, $\dpt_{R_0}[\piin,\piout] := \infty$. \>\>\> {\em \% no subsolution conforms to $\piin,\piout$}\\
  \\
  \> Otherwise, \>\>\>\>\>\>  {\em \% $R_0$ is a non-leaf dissection square}\\
  \>\>let $R_1, R_2, R_4, R_4$ denote the children of $R_0$.\\
\>\>For every valid configuration $(\piin_0,\piout_0)$ of $R_0$, initialize $\dpt_{R_0}(\piin_0,\piout_0):= \infty$.  \\
\>\>For every quintuple of indices $\left\{(\piin_i,\piout_i) \right\}_{i=0}^4$ to $\{\dpt_{R_i} \}_{i=0}^4$, \\
\>\>\> if $\left\{(\piin_i,\piout_i) \right\}_{i=0}^4$ are consistent,\\
\>\>\>\> $\dpt_{R_0}[\piin_0,\piout_0] :=\min \left\{ \dpt_{R_0}[\piin_0,\piout_0] , \sum_{i=1}^4 \dpt_{R_i}[\piin_i,\piout_i]\right\}$.\\
\end{tabbing}

\subsection{Running time}\label{sec:run-time}
Since each part of $\piin$ contains $O(\rrho)$ portals (since $\piin$ is compact), $F_P$ is a Steiner tree of $O(\rrho)$ terminals (portals and possibly one terminal) among the cells of $\piin$.  To avoid the cells that are not in $\piin$, we will require at $O(\ssigma^2)$ Steiner points.  $F_P$ can be computed in time proportional to $\ssigma$ and $\rrho$ (which are $O(1/\epsilon)$) by enumeration.
Since the number of compact configurations is polylogarithmic and
since there are $O(n\log n)$ dissection squares, the running time of
the dynamic program is therefore $O(n \log^\xi n)$, where $\xi$ is a
constant depending on $\epsilon$.

\subsection{Correctness (proof of Theorem~\ref{thm:correct})} \label{sec:dp-correct} 

We prove Theorem~\ref{thm:correct}, giving the correctness of our
dynamic program, by bottom-up induction.  In the following, we use the notation, definitions and conditions of {\sc populate}.  The base cases of the induction correspond to dissection squares that contain at most one terminal.
If any part $P$ of $\piin$ contains cells or portals that are disconnected, then there is no subsolution that is compatible with $\piin$ and $\dpt_{R_0}[\piin_0,\piout_0] = \infty$ represents this.  Otherwise 
the subsolution $F_0$ that is given by the union of $\{F_P\ : \ \text{part $P$ of $\piin$}\}$ is compatible with $\piin$ by construction.  Further $F_0$ satisfies the terminal property of conformance with $R_0$ by construction and the remaining properties since it is compatible with a valid conformation.

When $R_0$ contains more than one terminal, for a valid configuration $(\piin, \piout)$ of $R_0$, we must prove:
\begin{description}
\item[Soundness] If $\dpt_{R_0}[\piin_0,\piout_0]$ is finite then there is
  a subsolution $F_0$ that recursively conforms to $R_0$, is compatible with
  $(\piin_0,\piout_0)$ and whose length is $\dpt_{R_0}[\piin_0,\piout_0]$.
\item[Completeness] Any minimal subsolution $F_0$ that recursively conforms to
  $R_0$ and is compatible with $(\piin_0,\piout_0)$ has length at least
  $\dpt_{R_0}[\piin_0,\piout_0]$.
\end{description}
The proof of Theorem~\ref{thm:correct} follows directly from this.  We will use the following lemma:

\begin{lemma} \label{lem:compatible} Let
  $\{(\piin_i,\piout_i)\}_{i=0}^4$ be consistent configurations for
  dissection square $R_0$ and child dissection squares $R_1,\ldots,
  R_4$.  For $i = 1, \ldots, 4$, let $F_1,\ldots,F_4$ be subsolutions
  that recursively conform to $R_i$ and are compatible with
  $(\piin_i,\piout_i)$.  Then $\cup_{i=1}^4 F_i$ recursively conforms
  to $R_0$ and is compatible with $(\piin_0,\piout_0)$.
\end{lemma}

\begin{proof}
  Recall that $F_0$ is compatible with $ (\piin_0,\piout_0)$ if
  $\piin_0$ has one part for every connected component of $F_0$ that
  intersects $\partial R_0$ and that part consists of the cells and
  portals intersected by that component.  Consider a component $K$ of
  $F_0$ that intersects $\partial R_0$.  There must be a child
  dissection square $R_i$ with a part of $\piin_i$ that consists of
  the cells and portals intersected by $K \cap R_i$.  Consider all
  such parts $P_j$, $j = 1, \ldots$.  (Note that there may be more
  than one such part from a given child dissection square.)  These
  parts belong to a part $P$ of $\pi^\vee_0$.

  We argue that no other child configuration parts make up $P$.  For a
  contradiction, suppose another part $P'$ is in the make up of $P$.
  Since $(\piin_0,\piout_0)$ is consistent with the child
  configurations, $P'$ cannot share a cell with any of $P_j$, $j = 1,
  \ldots$ for otherwise $P$ would not survive the pruning given by the
  internal connectivity requirement of consistency.  Therefore, $P'$
  must share a portal with some $P_j$; the corresponding parts $K'$
  and $K_j$ would therefore also share this portal, implying that $K
  \cap K'$ is connected, a contradiction.

  Again, by the internal connectivity requirement of consistency, $P$
  is obtained from $P_j$, $j = 1, \ldots$ by:
  \begin{itemize}
  \item Removing the portals that are not in $R_0$.  The remaining
    portals are on $\partial R_0$, and $K$ connects them since $K_j$,
    $j = 1, \ldots$ connect them by the inductive hypothesis.
  \item Each cell $C$ of $P_j$ is replaced by the parent cell, which
    entirely contains $C \cap K$.
  \end{itemize}
  Finally $P$ is not removed altogether since $K$ intersects $\partial
  R_0$ and this intersection must contain a portal of $R_0$.
  Therefore, there is a part of $\piin_0$ obtained from $P$ that
  contains all the cells and portals intersected by $K$.
\end{proof}

\subsubsection*{Proof of soundness}

If $\dpt_{R_0}[\piin_0,\piout_0]$ is finite, then there must be entries
$\dpt_{R_i}[\piin_0,\piout_0]$ that are finite for $i = 1, \ldots, 4$ and such that $\dpt_{R_0}[\piin_0,\piout_0] = \sum_{i=1}^4 \dpt_{R_i}[\piin_0,\piout_0]$.
Then, by the inductive hypothesis, for $i = 1, \ldots, 4$, there is a
subsolution $F_i$ that recursively conforms to $R_i$, has length
$\dpt_{R_i}[\piin_i,\piout_i]$, and is compatible with
$\piin_i,\piout_i$.  We simply define $F_0 = \bigcup_{i=1}^4 F_i$; by
definition, $F_0$ has the desired length.  By Lemma~\ref{lem:compatible}, $F_0$ is compatible with $(\piin_0,\piout_0)$.  We show that $F_0$ conforms to $R_0$ by illustrating the four properties of conformance.

\paragraph{$\mathbf{F_0}$ satisfies the portal property} Let $K$ be a
component of $F_0 \cap \partial R_0$.  For some child $R_i$, the
intersection of $K$ with $\partial R_i \cap \partial R_0$ is nonempty.
Since $F_i$ satisfies the portal property, $K \cap \partial R_i
\cap \partial R_0$ must also contain a portal; that portal is also a
portal of $R_0$.

\paragraph{$\mathbf{F_0}$ satisfies the cell property} Let $C$ be a
cell of $R_0$ that is enclosed by child dissection square $R_i$.
Suppose for a contradiction that two connected components $K_1$ and
$K_2$ intersect both $C$ and $\partial R$.  Then $K_1 \cap R_i$ and
$K_2 \cap R_i$ must be connected components of $F_i$ that intersect
cells $C_1$ and $C_2$, respectively, and $\partial R_i$, where $C_1$
and $C_2$ are child dissection squares of $C$.  Since $F_i$ satisfies
the cell property w.r.t.\ $R_i$, $C_1 \neq C_2$ and these cells belong
to parts $P_1 \ne P_2$ of $\piin_i$.  By the internal connectivity
quirement of consistency, these cells would both get replaced by $C$,
implying that $\piin_0$ has two parts containing the same cell, a
contradiction.

\paragraph{$\mathbf{F_0}$ satisfies the terminal property}  
Consider a terminal $t$ in $R_i$ and $R_0$ such that $C_{R_i}[t]$ is in a part $P$ of $\piin_i$ (for otherwise, the terminal property follows from the inductive hypothesis).  If $t$'s mate is not in $R_0$, then, by the connecting property of valid configurations, $C_{R_0}[t]$ is in a part of $\piin_0$ and the terminal property follows from compatibility.  So suppose $t$'s mate, $t'$ is in $R_0$ (and child $R_j$).

Since the configurations are valid, $t'$ is in a part $P'$ of $\piin_j$.  If $\piout_0[C_{R_0}[t]] = \piout_0[C_{R_0}[t']]$, the terminal property follows from compatibility.  If not, then by the terminal connecting property of configuration consistency, either $\pi^\vee_0[C_{R_i}[t]] = \pi^\vee_0[C_{R_j}[t']]$. Since parts of child configurations cannot share cells, there must be a series of parts $P_1, \ldots, P_k$ where $P_1$ contains $C_{R_i}[t]$, $P_k$ contains $C_{R_j}[t']$ and parts $P_\ell$ and $P_\ell+1$ contain a common portal $p_\ell$ for $\ell = 1, \ldots, k-1$.  Since $F_1, \ldots, F_4$ are compatible with $\piin_1, \ldots, \piin_4$, respectively, by the inductive hypothesis, there is a component $K_\ell$ in $\cup_{i=1}^4 F_i$ that connects $t$ and $p_1$ (for $\ell=1$), $p_\ell$ and $p_{\ell+1}$ (for $\ell= 2, \ldots, k-1$) and $p_\ell$ to $t'$ (for $\ell = k$).  $\cup_{\ell = 1}^k K_\ell$ is a component in $F_0$ that connects $t$ and $t'$, giving the terminal property.

\paragraph{$\mathbf F_0$ satisfies the boundary property}
Since $(\piin_0,\piout_0)$ is a valid configuration, $\piin_0$ has at
most $4(\rrho+1)$ parts.  By compatibility, $F_0$ has at most
$4(\rrho+1)$ components intersecting $\partial R_0$.  This proves the
compactness property of conformance.

\subsubsection*{Proof of completeness}  

Let $\hat F_0$ be any minimal subsolution that recursively conforms to
$R_0$ and is compatible with $(\piin_0,\piout_0)$.  We show that $\hat
F_0$ has length at least $\dpt_{R_0}[\piin_0,\piout_0]$, proving
completeness.  For $i = 1, \ldots, 4$, let $\hat F_i = \hat F_0 \cap
R_i$.  \newcommand{\hpiin}{{\hat \pi}^{\text{in}}}
\newcommand{\hpiout}{{\hat \pi}^{\text{out}}} Since $\hat F_0$
recursively conforms to $R_0$, $\hat F_i$ recursively conforms to $R_i$.
For $i = 1, \ldots, 4$, let $(\hpiin_i, \hpiout_i)$ be a configuration
of $R_i$ that is compatible with $\hat F_i$.  By
Observation~\ref{obs:conform}, $(\hpiin_i, \hpiout_i)$ is a valid
configuration.  By the inductive hypothesis, $\leng(\hat F_i) \geq
\dpt_{R_i}[(\hpiin_i, \hpiout_i)]$.  It follows that $\leng(\hat F_0)
\geq \sum_{i=1}^4 \dpt_{R_i}[(\hpiin_i, \hpiout_i)]$.  If the child
configurations $\{(\hpiin_i, \hpiout_i)\}_{i=1}^4$ are consistent with
$(\piin_0,\piout_0)$, $\sum_{i=1}^4 \dpt_{R_i}[(\hpiin_i, \hpiout_i)]$
will be an argument to the minimization in {\sc populate} and
therefore $\leng(\hat F_0) \geq \dpt_{R_0}[\piin_0,\piout_0]$.  It is
therefore sufficient to show that the child configurations
$\{(\hpiin_i, \hpiout_i)\}_{i=1}^4$ are consistent with
$(\piin_0,\piout_0)$.  Equivalently, by Lemma~\ref{lem:compatible},
$\hat F_0$ is compatible with the configuration $(\hpiin_0,\hpiout_0)$
that is consistent with $\{(\hpiin_i, \hpiout_i)\}_{i=1}^4$ according
to the connectivity requirements of consistency.

\bigskip

\noindent This completes the proof of Theorem~\ref{thm:correct}.

\section{Proof of the Structure Theorem (Theorem~\ref{thm:structure})}\label{sec:struct}

In this section we give a proof of the Structure Theorem
(Theorem~\ref{thm:structure}).  We restate and reword the theorem
here for convenience.  It is easy to see that the statement here is
equivalent to the statement given in Section~\ref{sec:DP}; only the
terminal property of conformance is missing, but that is encoded by
feasibility.

\begin{reptheorem}{thm:structure}[Structure Theorem] \nonumber There is a feasible
  solution $F$ to the rounded Steiner forest problem having, in
  expectation over the random shift of the bounding box, length at
  most $\frac{2}{5} \epsilon \OPT$ more than $\OPT$ such that each
  dissection square $R$ satisfies the following three properties:
  
  \vspace{2mm}
  \begin{minipage}[r]{0.95\linewidth}
    \begin{description}
    \item[Boundary Property] For each side $S$, $F \cap S$
      has at most $\rrho$ non-corner components, where
      \begin{equation}
        \rrho=60\epsilon^{-1}\label{eq:rho}
      \end{equation}
    \item[Portal Property] Each component of $F \cap
      \partial R$ contains a portal.
    \item[Cell Property] For each cell $C$ of $R$, $F$ has at most one
      component that intersects both $\partial C$ and $\partial R$.
    \end{description}
  \end{minipage}
\end{reptheorem}

First, in a way similar to Arora, we illustrate the existence of a nearly-optimal
solution that crosses the boundary of each dissection square a small
number of times ({\em Boundary Property}) and does so at portals ({\em
  Portal Property}).  To that end, starting with the solution
$F_0$ as guaranteed by Lemma~\ref{lem:sum-of-crossings}, we augment
$F_0$ to create a solution $F_1$ that satisfies the Boundary
Components Property, then augment $F_1$ to a solution $F_2$ that also
satisfies the Portal Property.
The
{\em Cell Property} is then achieved by carefully adding to $F_2$ boundaries of cells
that violate the  Cell Property.

By Lemma~\ref{lem:sum-of-crossings}, $F_0$ is longer than $\OPT$ by
$\frac{\epsilon}{10}\OPT$.  We show that we incur an additional
$\frac{\epsilon}{10}\OPT$ in length in satisfying each of these three
properties, for a total increase in length of $\frac{4}{10} \epsilon
\OPT$, giving the Theorem.

\subsection{The Boundary Property}

We establish the Boundary Property constructively by starting with $F_1=F_0$ and adding closures
of the intersection of $F_1$ with the sides of dissection squares.
For a subset $X$ of a line, let $\text{closure}(X)$ denote the minimum
connected subset of the line that spans $X$.  For a side $S$ of a
dissection square $R$, a connected component of a subset of $S$ is a {\em
  non-corner component} if it does not include a corner of $R$
The construction is a simple greedy bottom-up procedure:

\begin{tabbing}
  {\sc SatisfyBoundary}:\\
  \quad \= For each $j$ decreasing from $\log L$ to 0,\\
  \> \quad  \= For each dissection line $\ell$ such that $\depth(\ell)\leq j$,\\
  \>\> \quad \= for each $j$-square with a side $S \subseteq \ell$, \\
  \> \>\> \quad \= if $|\{ \text{non-corner components of }F_1 \cap S\}| > \rrho$, \\
  \> \> \>\> \quad add $\text{closure}( \text{non-corner components of
  }F_1 \cap S)$ to $F_1$.
\end{tabbing}

\subsubsection*{{\sc SatisfyBoundary} establishes the
  Boundary Property}
Consider a dissection square $R$, a side $S$ of $R$, and the
dissection line $\ell$ containing $S$. The iteration involving $\ell$
and $j=\depth(\ell )$ ensures that, at the end of that iteration,
there are at most $\rrho$ components of $F_1\cap S$ not including the
endpoints of $S$, which are corners of $R$.  We need to show that
later iterations do not change this property.

Consider an iteration corresponding to $j'\leq j$, a line $\ell'$ with
$j'\geq \depth(\ell')$, and a side $S'\subseteq \ell'$ of a
$j'$-square $R'$.  By the nesting property and since $S'$ cannot be
enclosed by $S$, $S\cap \ell'$ is either empty, a corner of $R$ or
equal to $S$.  In the first case, $S \cap F_1$ is not affected by
adding a segment of $S'$.  In the second case, no new non-corner
component of $F_1\cap S$ appears.  In the third case, if adding a
segment of $S'$ would reduce $|S \cap F_1|$ to one.
See
Figure~\ref{fig:sat-bdy}. 

\begin{figure}[ht]
  \centering
  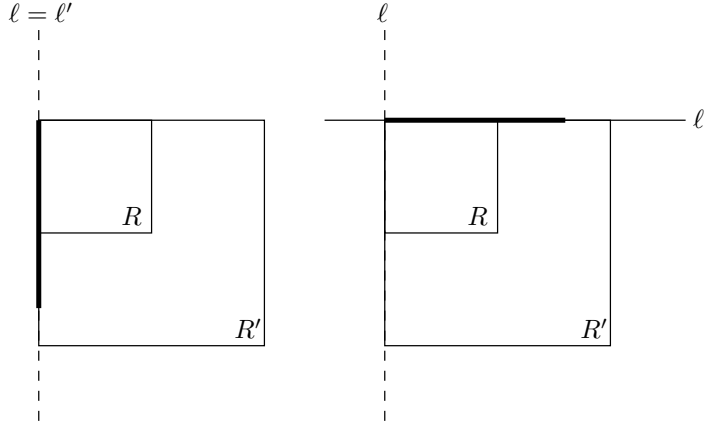
  \caption{The second (right) and third (left) cases for showing that
    {\sc SatisfyBoundary} can only decrease the number of components
    along the side of another dissection square or adding a corner
    component when a segement (thick line) of a dissection square side ($R\cap \ell$) is added to $F$ (not shown).}
  \label{fig:sat-bdy}
\end{figure}

\subsubsection*{The increase in length due to {\sc
    SatisfyBoundary} is small}\label{sec:leng-sat-bound}
For iteration $j$ of the outer loop and iteration $\ell$ such that $j\geq \depth(\ell)$ of the second
loop, let random variable $C_{\ell,j}$ denote the number of executions of the last step:\\
\centerline{ add $\text{closure}( \text{non-corner components of
  }F_1 \cap S)$ to $F_1$}\\
Note that, conditioning on $\depth(\ell)\leq j$,
$C_{\ell,j}$ is independent of $\depth(\ell)$ (however $C_{\ell,j}$ does depend on the random shift in the direction perpendicular to $\ell$). 
Initially the number of non-corner components of $F_1 \cap
\ell$ is at most the number of components, $|F_0 \cap \ell|$. 
As argued above: for every $j\geq \depth(\ell)$, every $j$-square either is disjoint from $\ell$ or has a side on $\ell$, so
dealing with a line $\ell'$ parallel to $\ell$ does not increase the number of components on $\ell$; For every $j<\depth(\ell)$, dealing with 
a line $\ell'$ perpendicular to $\ell$ can only introduce a corner component on $\ell$. 
So, the total number of non-corner components on $\ell$ never increases. Since it decreases by $\rrho$ at each of the $C_{\ell,j}$ 
closure operations, we have
$$ \sum_{j = \text{depth}(\ell)}^{\log L} C_{\ell, j} \leq |F_0
\cap \ell|/\rrho .$$
Since $\leng(S) = L/2^j$, the total increase in length resulting from
these executions is at most $C_{\ell,j} (L/2^j) $. Therefore, the
expected increase in length along $\ell$ is
\begin{eqnarray*}
E(\leng (F_1\cap\ell)-\leng (F_0\cap \ell))
  & \leq & \sum_i\Prob[\depth(\ell)=i]\sum_{j \geq i} E[C_{\ell,j}| \depth(\ell)=i] \frac{L}{2^j}\\
  & = & \sum_i \frac{2^i}{L} \sum_{j \geq i}  E[C_{\ell,j}| \depth(\ell)\leq j]  \frac{L}{2^j} \\
  & =    & \sum_j E[C_{\ell,j}| \depth(\ell)\leq j]\frac{1}{2^j} \sum_{i\leq j} 2^i\\
  & \leq & 2 E[\sum_{j\geq \depth (\ell)} C_{\ell,j}| \depth(\ell)]\\
  & \leq & 2 |F_0\cap \ell|/\rrho.
\end{eqnarray*}
Summing over all dissection lines $\ell$, and using the bounds on $\sum_\ell
|F_0 \cap \ell|$ and $\rrho$ as given by
Equations~\eqref{eq:crossings-vs-length}
and~\eqref{eq:rho},
respectively, we infer that the length of $F_1$ is at most
$\frac{\epsilon}{10}\OPT$ more than the length of $F_0$.
  
\subsection{The Portal Property}\label{subsection:establishing-the-portal-property}

We establish the Portal Property constructively by starting with $F_2=F_1$ and extending $F_2$
along the boundaries of dissection squares to nearest portals.  We say a component is portal-free if it does not contain
a portal.  The following construction  establishes the
Portal Property:

\begin{tabbing}
  {\sc SatisfyPortal}: \\
  \quad \= For each $j$ decreasing from $\log L$ to 0,\\
  \> \quad  \= For each dissection line $\ell$ such that $\depth(\ell)= j$,\\
  \> \> \quad \= for each portal-free component $K$ of $F_2 \cap \ell$,\\
  \> \> \> \quad\= extend $K$ to the nearest non-corner portal on $\ell$.
\end{tabbing}
  
\subsubsection*{{\sc SatisfyPortal} preserves the Boundary 
  Property} Focus on dissection line $\ell$. Before the iteration corresponding to $\ell$, 
possible extensions along lines $\ell'$ that are perpendicular to $\ell$ and of depth greater than of equal to $\depth (\ell)$
 do not extend to $\ell$, because $\ell'\cap\ell$ is a corner of $\ell'$. After the iteration corresponding to $\ell$, for each 
possible extension along lines $\ell'$ that are perpendicular to $\ell$ and of depth strictly less than  $\depth (\ell)$,
$\ell'\cap \ell$ is a corner of  any
dissection square $R$ with a side along $\ell$ containing $\ell\cap\ell'$, so the Boundary Property for $\ell$ is not violated.

\subsubsection*{The increase in length due to {\sc SatisfyPortal} is
  small} Consider a dissection line $\ell$.  When dealing with line $\ell$,  {\sc
  SatisfyPortal} only merges components  and,
in doing so, does not increase the number of components of $F_1 \cap
\ell$.  When dealing with a dissection line $\ell'$ perpendicular to $\ell$, 
As {\sc SatisfyPortal} might add  the
component $\ell \cap \ell'$  to $F_1 \cap \ell$.  However,
similar to the argument used above, in that case 
$\ell'\cap \ell$ is a corner of  any
dissection square $R$ with a side along $\ell$ containing $\ell\cap\ell'$. 
Since, by Lemma~\ref{lem:corner-portals},
corners are portals, no extension is made for this component.
Therefore, each component of $F_1 \cap \ell$ that does not already
contain a portal is an extension of what was originally already a component of $F_0 \cap
\ell$ and so, at most $|F_0 \cap \ell|$ extensions are made along $\ell$. 

Each of these extensions adds a length of at most
$L/(\aalpha 2^{\text{depth}(\ell)} )$ (the inter-portal distance for
line $\ell)$.  Therefore, the total length added along dissection line
$\ell$ is bounded by $|F_0 \cap \ell|\, L/(\aalpha 2^{\text{depth}(\ell)} )$.  Since $\Prob[\depth(\ell)=i] = 2^i/L$, the
expected increase in length due to dissection line $\ell$ is
\[
\sum_{i = 1}^{\log L} \frac{2^i}{L} |F_0\cap \ell| \frac{L}{2^i \aalpha} =
\frac{|F_0\cap \ell| \log L}{\aalpha}
\]
Summing over all dissection lines and using
Equations~\eqref{eq:crossings-vs-length}
and~\eqref{eq:inter-portal-distance}, we infer that the length of
$F_2$ is at most $\frac{\epsilon}{10}\OPT$ more than the length of
$F_1$.

\subsection{The Cell Property} \label{sec:cell-props}

We establish the Cell Property constructively by starting with $F_3 =
F_2$ and adding to $F_3$ boundaries of cells that violate the Cell
Property. 
Let $C$ be a cell of a dissection square $R$.  We say $C$ is {\em
  happy} with respect to the solution $F_3$ if there is at most one connected
component of $F_3$ that touches both the interior of $C$ and $\partial R$.
We cheer up an unhappy cell $C$ by adding to  $F_3$ a
subset $A$ of $\partial C$, as illustrated in Figure~\ref{fig:simple}:
\begin{equation}
  \label{eq:A}
  A(C,F_3) = \partial(C) \setminus \{\text{sides $S$ of }C\ : \ \depth(S)<\depth(C)\text{
    and } S\cap F_3 = \emptyset\}.
\end{equation}
Recall that each cell $C$ of $R$ is either coincident with a
dissection square that is a descendant of $R$ or is smaller than and
enclosed by a leaf dissection squares that is a descendant of $R$.
Definitions for the depth of a cell and its sides are inherited from
the definitions of dissection-square depths and dissection-line
depths.

\begin{figure}[ht]
\centering
    \includegraphics[scale=2]{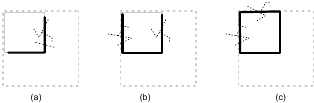}
    \caption{The three cases (up to symmetry) of augmenting $C$.  The
      dotted lines are $F_3$, $C$ is the smaller square and $C$'s
      parent is the larger square (to illustrate the relative depth of
      $C$'s sides).  In cases (a) and (b), the augmentation $A$ is not
      all of $\partial C$ so is open at the ends.  In (a), $F_3$
      intersects neither of the sides of $C$ that have depth less than
      that of $C$, so the augmentation $A$ consists only of the two
      sides having depth equal to that of $C$.  In (b), one of the
      low-depth sides intersects $F_3$, so it belongs to $A$.  In (c),
      both low-depth sides intersect $F_3$, so $A$ is all of $\partial
      C$.  }
    \label{fig:simple}
\end{figure}

\noindent Happiness of all cells, and therefore the Cell Property, is
established by the following procedure:
\begin{tabbing}
  {\sc SatisfyCellAbstract}:\\
  \quad \= While there is an unhappy cell $C$,\\
  \> \quad \= add $A(C,F_3)$ to $F_3$.
\end{tabbing}

Let $\cal C$ be the set of cells that we augment in the above
procedure.  

We claim that there is a function $h$ from
the cells $\cal C$ to the components of $F_0$ (the {\em original}
forest that we started with prior to the {\sc Satisfy} procedures) that is injective and, such
 that, for a cell $C$ of dissection square $R$,
$f(C)$ is a component of $F_0$ that intersects $\partial R$.

To define $h$, consider the following abstract directed forest $H$ whose vertices correspond to connected components of $F_0$ and whose edges correspond to augmentations made by {\sc SatisfyCell} (defined formally as follows).
An augmentation for cell $C$ is triggered
by the existence of at least two connected components $T,T'$ of the current $F_3$ that both touch the interior of $C$ 
and the boundary of its associated dissection square $R$.  Since the {\sc Satisfy} procedures augment the solution, $T$ and $T'$ each contain (at least one) connected component $T_0$ and $T_0'$ of $F_0$ -- it is the vertices corresponding to $T_0$ and $T_0'$ that are adjacent in $H$; we will show shortly that there exist such components that intersect $\partial R$.
Arbitrarily root each tree of $H$ and direct each of 
its edges away from the root. For augmentation of cell $C$, we then define $h(C)$ as the component of $F_0$ that corresponds to the head of the edge of $H$ associated with the augmentation of $C$.  Since each vertex of $H$ has indegree at most 1, $h$ is injective.

We show, by way of contradiction, that there is a component of $F_0$ contained by $T$ that intersects $\partial R$.  Consider all the components $\cal T$ of $F_0$ that are contained by $T$ and 
suppose none
of these intersect $\partial R$.  Let $\ell$  be a dissection line bounding $R$ that $T$ intersects.  Since $\cal T$ does not intersect $\partial R$, $T$ must have been created from $\cal T$ by augmentations (by way of {\sc SatisfyBoundary} and {\sc SatisfyCell}) one of which added a subset $X$ of dissection line $\ell'$ such that $X$ intersects $\ell$.  Since $\cal T$ does not intersect $X$ and neither {\sc SatisfyBoundary} nor {\sc SatisfyCell} augment to the corner of a dissection line, $\ell$ and $\ell'$ must be perpendicular.  Further $X$ is a subset of a side $S'$ of square $R'$ and does not contain a corner of $R'$.  In summary, $R$ and $R'$ are
dissection squares bounded by perpendicular dissection lines $\ell$
and $\ell'$ but for which $\ell \cap \ell'$ is not a corner of $R'$ or
$R$, contradicting that dissection squares nest.

We are now ready to give an implementation of {\sc SatisfyCellAbstract}:
\begin{tabbing}
  {\sc SatisfyCell}:\\
  \quad \= For each dissection line $\ell$,\\
  \>\quad \= for $j$ decreasing from $\log L$ to $\depth (\ell )$,\\
  \>\>\quad \= for each $j$-square $R$ with side $S \subseteq \ell$,\\
  \>\>\>\quad \= while there is an unhappy cell $C$ such that $h(C)$
  intersects $\ell$ \\
  \>\>\>\>\quad \= add $A(C,F_3)$ to $F_3$. 
\end{tabbing}
Since $h(C)$ intersects some side of some dissection square, this
procedure makes each of the cells happy.

\subsubsection*{The increase in length due to {\sc SatisfyCell} is
  small}

Let the random variable $C_{\ell,j}$ denote the number of augmentations corresponding to
dissection line $\ell$ and index $j$. Thanks to the injective mapping $h$, we have:
$$\sum_j C_{\ell,j}\leq |F_0\cap \ell|.$$

Since a cell has boundary length shorter than its $j$-square by a factor of $\ssigma$, 
the total increase in length corresponding
to these iterations is at most  $C_{\ell,j} \leng(j\text{-square})/\ssigma$.
Summing over $j$, the total length added by {\sc SatisfyCell} corresponding to
dissection line $\ell$  is at most
\[
\sum_{j \geq \depth(\ell)} C_{\ell,j} \frac{4L}{\ssigma 2^j}.
\]
Since the probability that grid line $\ell$ is a dissection line of
depth $k$ is $2^k/L$, the expected increase in length added by {\sc
  SatisfyCell} corresponding to
dissection line $\ell$  is at most
\[
\sum_k \frac{2^k}{L} \sum_{j \geq k} E[C_{\ell,j}| \depth (\ell)=k] \frac{4L}{\ssigma 2^j}.
\]
As  in
Section~\ref{sec:leng-sat-bound}, we observe that $C_{\ell,j}$ conditioned on $\depth (\ell)\leq j$ is independent of
$\depth (\ell)$. 
By the same swapping of sums as before, this is then bounded by 
$$(8/\ssigma) E[\sum_{j\geq \depth (\ell )} C(\ell, j)| \depth (\ell)]\leq \frac{8}{\ssigma}|F_0 \cap \ell|$$
Summing over all dissection lines, our bound on the expected
additional length becomes
\[
\frac{8}{\ssigma}\sum_\ell |F_0 \cap \ell|  = \frac{24}{\ssigma} (1+\epsilon) \OPT
\]
For $\ssigma = 240/\epsilon$, this is at most $\frac{\epsilon}{10}
\OPT$ by Equation~\eqref{eq:OPT-lb}.

\subsubsection*{{\sc SatisfyCell} maintains the Boundary and Portal
  Properties} 

We show that {\sc SatisfyCell} maintains the Boundary and Portal
Properties by showing that for any forest $F$ satisfying the Boundary
and Portal Properties, any single {\sc SatisfyCell} augmentation of
$F$ also satisfies these properties.

Let $C$ be an unhappy cell and let $R$ be a dissection square
satisfying the Boundary and Portal Properties.  Let $A$
be the augmentation that is used to cheer up $C$.  If $A \cap \partial
R$ contains a corner of $R$, then the Boundary Property is satisfied
because $A\cap \partial R$ would be a corner component and the Portal
Property is satisfied because the corners of dissection squares are
portals.

So, suppose that $A\cap \partial R$ is not empty but does not contain
a corner of $R$.  Refer to
Figure~\ref{fig:RandC} for relative positions of $R$ and $C$. Then $\partial C \cap \partial R$ cannot include an
entire side of $R$, so it must be that $\depth(C) > \depth(R)$.
Further, if $A\cap \partial R$ does not include a corner of $R$, then
$A \cap \partial R$ must be a subset of a single dissection line,
$\ell$.

\begin{figure}[ht]
  \centering
  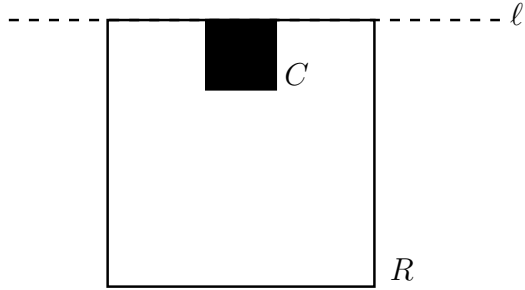
  \caption{Relative positions of $R$ and $C$.}
  \label{fig:RandC}
\end{figure}

If $A \cap \ell \cap F$ is not empty, then $F \cap \ell$ is not empty.
Since $F$ satisfies the Portal Property, $F \cap \ell$ also includes a
portal.  Since the addition of $A$ can only act to merge components, $|\ell \cal R \cap (F \cup A)| \le |\ell \cal R \cap A| $ and so $F$ still satisfies the Boundary Property.

If $A \cap \partial R \cap F$ is empty, then, by
Equation~\eqref{eq:A}, $\depth(\ell) \geq \depth(C)$.  But $\depth(C)
> \depth(R)$, so $\depth(\ell) > \depth(R)$.  This is impossible
because $\ell$ is a line bounding $R$.

This completes the proof of Theorem~\ref{thm:structure}.

\subsection{Proof of Theorem~\ref{thm:main}}

Recall Theorem~\ref{thm:main}, stating that there is a randomized $O(n
\polylog n)$-time approximation scheme for the Steiner forest problem
in the Euclidean plane.  The proof of this Theorem is a corollary of
Theorems~\ref{thm:structure},~\ref{thm:correct},~\ref{thm:modify2} and
Lemma~\ref{lemma:rounding} as follows.  Theorem~\ref{thm:correct}
guarantees that we can compute, using dynamic programming, a solution
that satisfies Theorem~\ref{thm:structure}.
Section~\ref{sec:run-time} argues that this DP takes $O(n \polylog n)$
time.  Lemma~\ref{lemma:rounding} and Theorem~\ref{thm:modify2} shows that we can convert the solution(s), of near-optimal cost, guaranteed by Theorem~\ref{thm:structure} to near-optimal solutions for the original problem, thus giving Theorem~\ref{thm:main}.

\section{Conclusion}

We have given a randomized $O(n \mathop{poly}\log n)$-time approximation
scheme for the Steiner forest problem in the Euclidean plane.
Previous to this result polynomial-time approximation schemes (PTASes)
have been given for subset-TSP~\cite{Klein06} and Steiner
tree~\cite{BKK07,BKM09} in planar graphs, using ideas inspired from
their geometric counterparts.  Since the conference version of this
paper appeared, a PTAS has been given for Steiner forest in planar
graphs by Bateni et~al.~\cite{BHM10}.  Like our result here, Bateni
et~al.\ first partition the problem and then face the same issue of
maintaining feasibility that we presented in
Section~\ref{sec:overview}, except in graphs of bounded treewidth.
They overcome this by giving a PTAS for Steiner forest in graphs of
bounded treewidth; they also show this problem in NP-complete, even in
graphs of treewidth 3.  Recently we have seen this technique generalized to prize collecting versions of the problem for both Euclidean and planar~\cite{BCEHKM11} instances.

\bibliographystyle{plain}
\bibliography{long,steiner_forest}

\end{document}